\documentclass[11pt,letter]{scrartcl}
\usepackage[noend]{algpseudocode}
\usepackage[hyphens]{url}
\usepackage[pdfborder={0 0 0}]{hyperref}
\usepackage{cite}
\usepackage{fullpage}
\usepackage{graphicx} 
\usepackage{caption}
\usepackage{subfigure}
\usepackage{amsmath,amssymb}
\usepackage{float}
\usepackage{todonotes}
\usepackage{complexity}
\usepackage[margin=3cm]{geometry}
\usepackage{xspace}
\usepackage{microtype}
\usepackage{pdfpages}
\usepackage{color}
\usepackage{amsthm}
\newtheorem{theorem}{Theorem}
\newtheorem{lemma}[theorem]{Lemma}
\newtheorem{definition}[theorem]{Definition}

\newtheorem{algo}[theorem]{Algorithm}

\newcommand{\size}{\operatorname{size}}
\newcommand{\Rph}{R_{PH3}^\infty}

\newcommand{\ignore}[1]{}

\title{Parallel Online Algorithms \\ for the Bin Packing Problem}
\author{S\'andor P.~Fekete, Jonas Grosse-Holz, Phillip Keldenich, Arne Schmidt\\
	\small Department of Computer Science, TU Braunschweig, Germany.\\ \small \{s.fekete, j.grosse-holz, p.keldenich, arne.schmidt\}$@$tu-bs.de}
\date{}
\begin{document}
	\maketitle
	\begin{abstract}
	We study \emph{parallel} online algorithms:
	For some fixed integer $k$, a collective of $k$ parallel processes that perform
	online decisions on the same sequence of events forms a $k$-\emph{copy algorithm}. 
	For any given time and input sequence, 
	the overall performance is determined by the best of the $k$ individual total results. 
	Problems of this type have been considered for online makespan minimization;
	they are also related to optimization with \emph{advice} on future events, i.e.,
	a number of bits available in advance.

	We develop \textsc{Predictive Harmonic}$_3$ (PH3), a relatively simple
	family of $k$-copy algorithms for the online Bin Packing Problem, whose joint
	competitive factor converges to 1.5 for increasing $k$. In particular,
	we show that $k=6$ suffices to guarantee a factor of $1.5714$ for PH3, 
	which is better than $1.57829$, the performance of the best known 1-copy algorithm
	\textsc{Advanced Harmonic}, while $k=11$ suffices to achieve a factor of $1.5406$,
	beating the known lower bound of $1.54278$ for a single online algorithm. 
	In the context
	of online optimization with advice, our approach implies that 4 bits suffice to achieve 
	a factor better than this bound of $1.54278$, which is considerably less than the previous bound of 15 bits.
	\end{abstract}
\section{Introduction}
\newcommand{\hm}{\textsc{Harmonic}\textsubscript{M}}

When dealing with unknown future events, optimization with incomplete information typically
considers the competitive factor of an online algorithm as its performance measure;
the objective becomes to develop a single strategy that performs reasonably well against the worst case. 
This focus on just \emph{one} option is more restrictive than 
hedging strategies in a wide variety of other scientific and application fields; these
typically make use of \emph{several} parallel 
choices, thereby increasing the chance that one of them
will yield satisfactory results. Examples include scenarios from 
biology, where a large and diverse progeny increases the odds of surviving offspring;
finance and insurance, where a suitable combination of investment strategies is employed to 
balance a portfolio against extreme losses;
and engineering, where redundancy is used to protect against catastrophic failure,
either on individual components (such as parts in a machine) or on whole systems (such 
as automata in a robot swarm or spacecraft in a group of satellites), where it suffices
that just one machine delivers a good outcome.

In this paper, we consider such \emph{parallel} online strategies: Instead of making
a single sequence of decisions, we consider $k$ parallel processes for some fixed integer
$k$, which we call a $k$\emph{-copy algorithm}; the objective is to make the best of 
these $k$ outcomes as good as possible, even in the worst case.
We demonstrate the potential of this approach for the well-studied Bin Packing Problem,
for which it is known that no single deterministic online algorithm can achieve a competitive factor
below $1.5401$.

\subsection{Our Results}
We define a family of $k$-copy algorithms for the online Bin Packing Problem,
called \textsc{Predictive Harmonic}$_3$ (PH3),
whose asymptotic competitive ratio converges to $1.5$ for large $k$.
We show that $k=6$ suffices to guarantee a factor of $1.5714$, 
which is better than $1.57829$, the performance of the best known 1-copy algorithm
\textsc{Adanved Harmonic}~\cite{balogh2018algorithm}. Moreover, $k=11$ suffices to achieve a competitive ratio of $1.5406$
beating the known lower bound of $1.54278$ for a 1-copy algorithm~\cite{balogh2018bound}.
In the context of online optimization with advice, our approach implies that 4 bits 
suffice to achieve less than $1.5401$, which is considerably less than the previous bound of 16 bits
of \textsc{RedBlue} by Angelopoulos et al.~\cite{Angelopoulos2015}; in fact, 
for $k=16$ (corresponding to four bits of advice) 
PH3 achieves a ratio of $1.5305$, compared to $3.3750$ for \textsc{RedBlue}, while 
$k=65,536$ (corresponding to 16 bits of advice) yields a factor of $1.5001$ for PH3,
but $1.5293$ for \textsc{RedBlue}.

\subsection{Related Work on Online Bin Packing}
There is a wide range of online algorithms for bin packing.
The Next Fit algorithm~\cite{Csirik1998} achieves a competitive ratio of 2, whereas ``Almost Any Fit'' algorithms~\cite{JOHNSON1974272} like First Fit or Best Fit achieve competitive ratios of $1.7$.

\ignore{An algorithm considered very early is the \textsc{Next Fit} (NF) algorithm.
NF packs each item into the so-called \emph{active} bin.
Once an item does no longer fit into the active bin, the bin is closed and a new empty bin is used as active bin.
NF achieves an asymptotic competitive ratio of $R_{NF}^\infty = 2$ \cite{Csirik1998}.

The family of \textsc{Almost Any Fit} (AAF) algorithms, as introduced by Johnson \cite{JOHNSON1974272}, is characterized by only packing an item into the bin with the lowest level, if there is more than one bin with that level or that bin is the only one that has enough room.
The most important algorithms in this class are the \textsc{First Fit} (FF) and the \textsc{Best Fit} (BF) algorithms.
FF gives each used bin an index.
Each item is packed in the bin with the lowest index in which it fits.
BF packs each item in the bin with the highest level in which the item fits.
Each AAF algorithm $A$ achieves an asymptotic competitive ratio $R_A^\infty = 1.7$.
}

An important online bin packing algorithm is \hm{}, introduced by Lee and Lee \cite{Lee:1985:SOB:3828.3833},
which achieves a competitive ratio of less than $1.692$ for $M \to \infty$.
Based on \hm{}, \textsc{Son Of Harmonic} by Heydrich and van Stee~\cite{DBLP:journals/corr/HeydrichS15} achieves a competitive ratio of $1.5816$.
The currently best known algorithm is \textsc{Advanced Harmonic}, which achieves a competitve ratio of $1.57829$~\cite{balogh2018algorithm}.

For lower bounds, Yao~\cite{Yao:1980:NAB:322186.322187} established a value of
$3/2$ that was later improved to $1.536$, independently by Brown~\cite{brown1979} and by
Liang~\cite{LIANG198076}.
Using a generalization of their methods,
van Vliet~\cite{VANVLIET1992277} proved a lower bound of $1.5401$. 
Balogh et al.~\cite{balogh2018bound} improved the lower bound to $1.54278$.

\subsection{Related Work on Online Bin Packing with Advice}
In the context of online algorithms with advice, Boyar et al.\ \cite{Boyar2016} showed that an online algorithm with $n\lceil \log(OPT(I)) \rceil$ bits of advice is sufficient and that at least $(n - 2 OPT(I))\cdot \log (OPT(I))$ bits of advice are necessary to achieve optimality.
In the same paper, they presented an online bin packing algorithm, namely \textsc{ReserveCritical}, with $O(\log(n)) + o(\log(n))$ bits of advice that is $1.5$-competitive and an algorithm with $2n+o(n)$ bits of advice that is $\frac 4 3$-competitive.
Zhao and Shen~\cite{zhao2014advice} developed an algorithm using $3n+o(n)$ bits of advice achieving a competitive ratio of $\frac 5 4\textsc{opt} + 2$.
Renault et al.~\cite{renault2015online} developed an $(1+\varepsilon)$-competitive algorithm using $O(\frac 1 \varepsilon \log \frac 1 \varepsilon)$ bits of advice per request.

Based on \textsc{ReserveCritical}, Angelopoulos et al.~\cite{Angelopoulos2015} developed the algorithm \textsc{RedBlue} with constant advice that is $1.5$-competitive.
Their second algorithm achieves a competitive ratio of $1.47012+\varepsilon$ with finite advice that is exponentially dependent of $\varepsilon$.
However, to beat the competitive ratio of $1.5$ already an enormous amount of advice is needed, which makes the algorithm impractical.

In terms of lower bounds, Boyar et al.\ \cite{Boyar2016} proved that no competitive ratio better than $9/8$ can be reached by any algorithm that uses sub-linear advice.
Angelopoulos et al.\ \cite{Angelopoulos2015} improved this bound to $7/6$.

\subsection{Related Work on Parallel Online Algorithms}
Parallel algorithms have already been considered in the field of online algorithms with advice.
Boyar et al.\ \cite{Boyar2014} presented an algorithm for the online list
update problem, making use of $2$ bits of advice to choose one out of three
algorithms.  This algorithm achieves a competitive ratio of $5/3$, beating the
lower bound for conventional online algorithms of $2$.
A practical application of this algorithm was shown by Kamali and Ortiz \cite{6824445}, who applied it in the Burrows-Wheeler transform compression.
More work on parallel online algorithms include parallel scheduling~\cite{albers2017online}, finding independent sets~\cite{halldorsson2002online} and the ``multiple-cow'' version of the linear  search problem~\cite{lopez2004line}.

While online algorithms with advice mostly focus on the amount of advice to
allow classification of online algorithms and problems, $k$-copy online
algorithms focus on small finite values for $k$ and thus small finite amounts
of advice, with more emphasis on practical application.  The perspective on
different algorithms running in parallel instead of abstract arbitrary
information facilitates finer optimization in some cases. 

Also, when considering online algorithms with advice, the number of algorithm can only be doubled by increasing the amount of advice by one bit.
The perspective of k-copy algorithms allows arbitrary $k \in \mathbb{N}$ for the number of algorithms.

\section{Preliminaries}
\label{c:preliminaries}

\subsection{\boldmath{$k$}-Copy Online Algorithms}
In this paper, we consider $k$ online algorithms $A_1,\dots, A_k$, each of them processing the same input list $I$ in parallel. 
We call the set $\mathcal{A}:=\{A_1,\dots,A_k\}$ a \emph{$k$-copy online algorithm}.

For an input list $I$ and an online algorithm $A$, let $A(I)$ denote the number of bins used by $A$ and $\textsc{opt}(I)$ denote the number of bins used in an optimal offline solution.
The absolute competitive ratio $R_{\mathcal{A}}$ for a $k$-copy online algorithm $\mathcal A$ is defined as
\begin{align*}
R_\mathcal A = \sup_I \left\{ \frac{\min_{A \in \mathcal{A}} A(I)} { \textsc{opt}(I)} \right \}.
\end{align*}
The asymptotic competitive ratio $R_\mathcal{A}^\infty$ for algorithm $\mathcal A$ is defined as
\begin{align*}
R_\mathcal A^\infty = \lim_{n \to \infty} \sup_I \left\{\frac{\min_{A \in \mathcal{A}} A(I)} { \textsc{opt}(I)} \ \middle |\ \textsc{opt}(I) = n \right\}
\end{align*}

As already stated by Boyar et al.~\cite{boyar2017online}, any $k$-copy online algorithm can be converted into an online algorithm with advice, and vice versa.

\begin{lemma}\label{l:kcopy_advice}
	
	Any $k$-copy online algorithm can be converted into an online algorithm with $l=\lceil \log_2(k) \rceil$ bits of advice that achieves the same competitive ratio.
	Conversely, any online algorithm with $l \in \mathbb{N}$ bits of advice can be converted into a $k$-copy online algorithm without advice with $k=2^l$ that achieves the same competitive ratio.
\end{lemma}

\begin{proof}
	Let $\mathcal{A} = \{ A_1, A_2, \dots, A_k \}$ be a $k$-copy algorithm.
	Construct the online algorithm $A'$ that gets a value $i \in \{ 1, 2, \dots, k \}$ as advice, specifying the index $i$ of the algorithm $A_i \in \mathcal{A}$ that performs best on the given input sequence.
	The value $i$ can be encoded using $\lceil \log_2(k) \rceil$ bits.
	$A'$ then behaves like $A_i$ and thus achieves the same competitive ratio as $\mathcal{A}$.
	
	Let $A$ be an online algorithm that gets $l \in \mathbb{N}$ bits of advice.
	Construct the online $k$-copy algorithm $\mathcal{A}'$ with $k = 2^l$ algorithms $A_i, i \in \{ 1, 2, \dots, k \}$.
	For each $i \in \{ 1, 2, \dots, k \}$, the algorithm $A_i$ behaves like $A$ given $i$ encoded in binary as advice.
	As the values $i \in \{ 1, 2, \dots, k \}$ cover every possible configuration of the advice bits, for any advice given to $A$, there is an algorithm $A_i \in \mathcal{A}'$, that assumes this advice.
	Accordingly, there is an algorithm $A_i \in \mathcal{A}'$, that performs as well as $A$, i.e., the best algorithm $A_j \in \mathcal{A}$ that performs at least as well as $A$.
	Thus, $\mathcal{A}'$ performs at least as well as $A$.
\end{proof}

\subsection{Bin Packing}
In the online version of bin packing, we are given a list of items $I:=\langle a_1,\dots, a_n\rangle$ with $a_i \in (0,1]$ for $i\in \{1,\dots,n\}$.
These items must be packed by an algorithm, one at a time,  without any information on subsequent items and without the possibility to change previous decisions. 
The goal is to pack all items into a minimum number of bins with unit capacity.

\begin{definition}[Item size]
	
Let $S=\left[0,\frac{1}{3}\right]$,
$M=\left(\frac{1}{3},\frac{1}{2}\right]$,
$L=\left(\frac{1}{2},\frac{2}{3}\right)$ and 
$XL=\left[\frac{2}{3},1\right]$.
We call items in $S$ \emph{small}, items in $M$ \emph{medium}, items in $L$ \emph{large} and items in $XL$ \emph{extra large}.
For a list $I = \langle a_1, a_2, \dots a_n \rangle$, the set of items $Set(I) \cap XL$ is noted as $I_{XL}$ for improved readability. The subsets $I_L$, $I_M$ and $I_S$ are used analogously.
\end{definition}

\begin{definition}[Size function]
	Let $S$ be a set (or list) of items.
	Then, $\size(S) := \sum_{i \in S}i$.
	For a bin $b$, we refer to $\size(b)$ as the size of the bin, i.e., the sum of items already packed in $b$.

\end{definition}

\begin{definition}[Sub-bins]

Given a bin $b$, it can be split into two parts $b_{1}$ and $b_{2}$, such that the sum of their capacities is equal to the capacity of $b$.
We refer to $b_{1}$ and $b_{2}$ as sub-bins.
We call a sub-bin with capacity $C$ a $C$-sub-bin.

\end{definition}

As sub-bins are not packed with an amount larger than their capacity, each sub-bin can be packed independently from the other.

\section{\textsc{Predictive Harmonic$_3$}}

Now we introduce the algorithm \textsc{Predictive Harmonic$_3$} (PH3).
Although developed independently, it bears many similarities to \textsc{ReserveCritical} and \textsc{RedBlue}.
PH3 uses the same classifications as the other two algorithms and tries to pack all large items with small items, such that the corresponding bins are packed to a level of at least $2/3$.
However, in contrast to \textsc{RedBlue}, 
the information needed by PH3 does not depend on the result of \textsc{ReserveCritical}, 
but only on the number and size of certain item types, and can be calculated in linear time.

The main idea of PH3 is to guess the ratio of how many small items must be packed with large items to obtain a packing density of $2/ 3$.
Having multiple instances of PH3, every instance can guess a different ratio to get close to a competitive ratio of $1.5$.

\begin{algo}{\textsc{Predictive Harmonic$_3$}}
\normalfont

Given a list $I= \langle a_1,a_2,\dots,a_n \rangle$ of items $a_i \in (0,1], i \in {1,\dots,n}$, and a ratio $r_L \in [0,1]$,
the algorithm packs the items as follows:
\begin{itemize}
	\item Extra large items are packed into individual bins. 
	These bins are called XL-bins, the set of all XL-bins is called $B_{XL}$.
	
	\item Large items are packed into individual bins.
	These bins are called L-bins, the set of all L-bins is called $B_L$.
	Furthermore, we split each L-bin into a $\frac 2 3$-sub-bin (for large items) and a $\frac 1 3$-sub-bin (for small items).
	
	\item Medium items are packed into separate bins together with other medium items (note that at most two of them fit into one bin).
	These bins are called M-bins, the set of all M-bins is called $B_M$.
	
	\item Small items are packed into a $\frac 1 3$-sub-bin of L-bins in a next fit manner, if the size of small items packed into L-bins is smaller than $r_L$ times the total size of small items packed so far; otherwise we pack the small item into S-bins.
\end{itemize}

\end{algo}

\subsection{Competitive Ratio}
Before we proceed to prove the asymptotic competitive ratio, we recall some lower bounds for an optimal offline solution.

\begin{lemma}{(Lower bounds on OPT)}
	\label{l:lbound_OPT}
	For any input sequence $I$ in online bin packing, the following lower bounds on $OPT$ hold:
	\begin{enumerate}
		\item $OPT \geq |I_{XL}| + |I_L| \geq |I_L|$
		\item $\displaystyle OPT \geq |I_{XL}| + \frac{|I_M|+|I_L|}{2}$
		\item $OPT \geq \size(I)$
	\end{enumerate}
\end{lemma}

\begin{proof}\
	
	\begin{enumerate}
		\item As any two items $a,a' \in XL \cup L$ with $a \neq a'$ fulfill $a+a' > 1$, OPT has to put each of these items in a different bin. Thus, $OPT \geq |I_{XL}|+|I_L|$.
		\item First, note that no item $a \in M \cup L$ can be packed with an item $a' \in XL$.
		Also, no more than two items $a, a' \in M \cup L$ can be packed in the same bin, as each of these items is greater than $1/3$.
		Thus, 
		\begin{align*}
		OPT & \geq |I_{XL}|+\frac{|I_M|+|I_L|}{2}
		\end{align*}
		\item As OPT has to pack the items into unit-sized bins, the number of bins must be at least the items total size $\sum_{a \in I}a = \size(I)$.\qedhere 
	\end{enumerate}
\end{proof}

Using this bounds and performing a case analysis, we can prove the following theorem. 

\begin{theorem}
	\label{t:cr}
	Let $r_L^* = min\left\{\frac{|I_L|}{6\;\size(I_S)}, 1 \right\}$\footnote{The intuition of this value is that at least $1/2$ of each $1/3$-sub-bin must be filled to guarantee a packing density of $2/3$. Therefore, for $|I_L|$ bins, we have to fill up a total capacity of $\frac {|I_L|}{6}$ with small items.}
	and $\delta = r_L - r_L^*$.
	\ignore{
		\begin{align*}
		r_L^* = min\left\{\frac{|I_L|}{6\;\size(I_S)}, 1 \right\} \text{ and } \delta = r_L - r_L^*
		\end{align*}
	}
	PH3 achieves the asymptotic competitive ratio
	\begin{align*}
	R_{PH3}^\infty
	\leq 
	\begin{cases}
	\dfrac{3}{2} + min\left\{\dfrac{1}{4r_L^*}, \dfrac{3}{6 r_L^* + 2}\right\} (-\delta) &\text{ for } \delta \leq 0
	\vspace{2mm}\\
	\dfrac{3}{2} + min\left\{\dfrac{3}{4r_L^*}, \dfrac{9}{6 r_L^* + 2}\right\} \delta &\text{ for } \delta \geq 0.\\
	\end{cases}
	\end{align*}
\end{theorem}

\begin{proof}
	First, consider $r_L^*=\frac{|I_L|}{6 \; \size(I_S)}$.
	
	\begin{align*}
	PH_3 &= |B_{XL}| + |B_M| + |B_S| + |B_L|
	\end{align*}
	
	As each extra large item $a \in I_XL$ is packed into a separate XL-bin and each such item satisfies $a \geq 2/3$, we get
	\begin{equation*}
	|B_{XL}| = |I_{XL}| \leq \frac{3}{2} \; \size(I_{XL}).
	\end{equation*}
	
	Two medium items are packed into each M-bin, except for the last one if the number of M-items is odd.
	Each M-item $a \in I_M$ fulfills $a > 1/3$, thus
	\begin{equation*}
	|B_M|=\left\lceil\frac{|I_M|}{2}\right\rceil \leq \frac{|I_M|+1}{2} < \frac{3}{2} \size(I_M) + \frac{1}{2}.
	\end{equation*}
	
	PH3 will only open a new S-bin, if the current item $a \in I_S$ is supposed to be packed into an S-bin and does not fit into the currently open bin.
	As $\forall a \in I_S : i \leq 1/3$, any S-bin $b \in B_S$ except the last one has to fulfill $\size(b) > 2/3$, otherwise another item $a \in I_S$ would fit.
	A fraction of $(1 - r_L)$ of the size of small items plus at most one item is packed into S-bins.
	If a bin $b \in B_S$ with $\size(b) < 2/3$ exists, this additional item can be packed there, otherwise it is packed into an additional bin.
	In either case, there is at most one bin $b \in B_S$ with $\size(b) < 2/3$.
	\begin{align*}
	|B_S| & \leq \frac{3}{2} \; (1-r_L) \; \size(I_S) + 1\\
	& =    \frac{3}{2} (1-\delta)\; \size(I_S) - \frac{3}{2} \; r_L^* \size(I_S) + 1\\
	& =    \frac{3}{2} (1-\delta)\; \size(I_S) - \frac{3}{2} \frac{|I_L|}{6 \; \size(I_S)}\; \size(I_S) + 1\\
	& =    \frac{3}{2} (1-\delta)\; \size(I_S) - \frac{1}{4}\; |I_L| + 1
	\end{align*}
	
	For $B_L$, consider the subsets $B_{LL} = \{ b \in B_L |\, b \text{ contains a large item}\}$ and $B_{LS} = \{ b \in B_L |\, b \text{ contains a small item}\}$.
	As both small and large items are packed into L-bins by increasing index, either $B_{LS} \subseteq B_{LL}=B_L$ or $B_{LL} \subseteq B_{LS}=B_L$ holds.
	Thus, $|B_L| = max\{|B_{LL}|, |B_{LS}|\}$.
	
	No two large items are packed into the same bin, so $|B_{LL}|=|I_L|$.
	
	Let $I_S'$ the set of small items packed into L-bins.
	As small items are smaller than $1/3$ and are only packed into L-bins if less than a part $r_L$ of the small items packed so far have been packed into L-bins, $\size(I_S') < r_L \; \size(I_S) + 1/3$.
	The set $I_S'$ is packed into the $1/3$-sub-bins by a NF-algorithm.
	For an optimal packing $OPT'$ of small items in these sub-bins, NF$\leq 2 \; OPT' + 1$.
	Each sub-bin has capacity $1/3$, so OPT needs at least $3 \size(I_S')$ bins to pack all items.
	\begin{align*}
	|B_{LS}| & \leq 2 \; OPT' +1\\
	& \leq 6 \; \size(I_S') + 1\\
	& < 6 \; r_L \; \size(I_S) + 3\\
	& = 6 \; (r_L^* + \delta) \size(I_S) + 3\\
	& = 6 \; \left(\frac{|I_L|}{6 \; \size(I_S)} + \delta \right) \size(I_S) + 3\\
	& = |I_L| + 6 \; \delta \; \size(I_S) +3
	\end{align*}
	
	The maximum of the bounds on $|B_{LL}|$ and $|B_{LS}$ is an upper bound on $|B_L|$.
	
	\begin{align*}
	|B_L| & \leq max\{|I_L|, |I_L| + 6 \; \delta \; \size(I_S) +3\}\\
	& = |I_L| + max\{0, 6 \; \delta \; \size(I_S) + 3\}
	\end{align*}
	
	Combining the bounds for $|B_S|$ and $|B_L|$ yields
	
	\begin{align*}
	|B_S| + |B_L|  & \leq 
	\underbrace{\frac{3}{2} (1-\delta)\; \size(I_S) - \frac{1}{4}\; |I_L| + 1}_{\geq |B_S|}
	+ \underbrace{|I_L| + max\{0, 6 \; \delta \; \size(I_S) + 3\}\vphantom{\frac{1}{4}}}_{\geq |B_L|}\\
	& =  \frac{3}{2} (1-\delta)\; \size(I_S) + \frac{3}{4}\; |I_L| + 1 + max\{0, 6 \; \delta \; \size(I_S) + 3\}
	\end{align*}
	
	Each large item $a \in I_L$ is greater than $1/2$. Therefore $|I_L| < 2 \; \size(I_L)$.
	\begin{align*}
	|B_S| + |B_L|  & <  \frac{3}{2} (\size(I_S) + \size(I_L)) -\frac{3}{2} \delta \size(I_S) + 1 + max\{0, 6 \; \delta \; \size(I_S) + 3\}
	\end{align*}
	
	Let
	\begin{align*}
	\Delta        & = \frac{3}{2} + max\{0 , 6\; \delta \; \size(I_S) + 3\} - \frac{3}{2} \delta \; \size(I_S)\\
	& = max\left\{\frac{3}{2} - \frac{3}{2} \delta \; \size(I_S), \frac{9}{2} +  \frac{9}{2} \delta \; \size(I_S)\right\}\\
	\Rightarrow |B_S| + |B_L| & < \frac{3}{2} (\size(I_S) + \size(I_L)) + \Delta -\frac{1}{2}
	\end{align*}
	Summing up,
	\begin{align*}
	PH_3
	& = |B_{XL}| + |B_M| + |B_S| + |B_L|\\
	& < \underbrace{\frac{3}{2} \; \size(I_{XL})}_{\geq |B_{XL}|}
	+ \underbrace{\frac{3}{2} \size(I_M)+\frac{1}{2}}_{> |B_M|}
	+ \underbrace{\frac{3}{2} (\size(I_S) + \size(I_L)) + \Delta - \frac{1}{2}}_{\geq |B_S| + |B_L|}\\
	& = \frac{3}{2} (\size(I_{XL}) + \size(I_M) + \size(I_S) + \size(I_L)) + \Delta\\
	& = \frac{3}{2} \size(I) + \Delta.
	\end{align*}
	
	Because of Lemma~\ref{l:lbound_OPT}, $OPT \geq \size(I)$.
	\begin{align*}
	\Rightarrow PH_3 & < 
	\frac{3}{2} OPT + \Delta
	\end{align*}
	
	This yields the asymptotic competitive ratio:
	\begin{align*}
	\lim_{OPT \to \infty} \frac{PH_3}{OPT} & \leq 
	\lim_{OPT \to \infty} \frac{3}{2} + \frac{\Delta}{OPT}\\
	& = \lim_{OPT \to \infty} \frac{3}{2} +  \frac{max\left\{\frac{3}{2} - \frac{3}{2} \delta \; \size(I_S), \frac{9}{2} +  \frac{9}{2} \delta \; \size(I_S)\right\}}{OPT}\\
	& = \lim_{OPT \to \infty} \frac{3}{2} + max\left\{-\frac{3}{2} \delta,\frac{9}{2} \delta\right\} \; \frac{\size(I_S)}{OPT}
	\end{align*}	
	Clearly, $OPT$ is an upper bound on $\size(I_S)$. 
	However, better bounds can be found using $r_L^*$.
	\begin{align*}
	r_L^* = \frac{|I_L|}{6 \size(I_S)}
	&&\Leftrightarrow&&
	\size(I_S) = \frac{|I_L|}{6 r_L^*}
	\end{align*}
	Because of Lemma~\ref{l:lbound_OPT}, $|I_L|$ is a lower bound on $OPT$.
	\begin{align*}
	\Rightarrow \size(I_S) = \frac{|I_L|}{6 r_L^*} \leq \frac{OPT}{6 r_L^*}
	\end{align*}
	As each item $i \in I_L$ is larger than $1/2$, $|I_L| < 2 \size(I_L)$.
	According to Lemma~\ref{l:lbound_OPT}, $OPT \geq \size(I)$, i.e. $OPT \geq \size(I_S)+\size(I_L)$.
	\begin{align*}
	\Rightarrow \size(I_S)                                      &= \frac{|I_L|}{6 r_L^*} \leq \frac{\size(I_L)}{3 r_L^*}\\
	\Leftrightarrow \left(1+\frac{1}{3 r_L^*}\right) \size(I_S) &\leq  \frac{\size(I_L)+\size(I_S)}{3 r_L^*}\\
	\Leftrightarrow \size(I_S)                                  &\leq \frac{\size(I_L)+\size(I_S)}{3 r_L^*+ 1}\\
	&\leq \frac{OPT}{3 r_L^*+ 1}
	\end{align*}
	As these are both upper bounds on $\size(I_S)$, their minimum is also an upper bound on $\size(I_S)$.
	This results in an asymptotic upper bound on $\Delta$ in comparison to $OPT$:
	\begin{align*}
	\lim_{OPT \to \infty} \frac{\Delta}{OPT} &=
	max\left\{-\frac{3}{2} \delta,\frac{9}{2} \delta\right\} min\left\{\frac{1}{6 r_L^*}, \frac{1}{3 r_L^* + 1} \right\}\\
	\Leftrightarrow \lim_{OPT \to \infty} \frac{\Delta}{OPT} &=
	\begin{aligned}
	\begin{cases}
	min\left\{\dfrac{1}{4r_L^*}, \dfrac{3}{6 r_L^* + 2}\right\} (-\delta) &\text{ for } \delta \leq 0
	\vspace{2mm}\\
	min\left\{\dfrac{3}{4r_L^*}, \dfrac{9}{6 r_L^* + 2}\right\} \delta &\text{ for } \delta \geq 0
	\end{cases}
	\end{aligned}
	\end{align*}
	This results in a competitive ratio dependent on $\delta$, proving \autoref{t:cr} for $r_L^*=\frac{|I_L|}{6 \; \size(I_S)}$:
	\begin{align*}
	R_{PH3}^\infty 
	&= \lim_{OPT \to \infty} \frac{PH_3}{OPT} \\
	&\leq \lim_{OPT \to \infty} \frac{3}{2} + \frac{\Delta}{OPT} \\
	&=
	\begin{cases}
	\dfrac{3}{2} + min\left\{\dfrac{1}{4r_L^*}, \dfrac{3}{6 r_L^* + 2}\right\} (-\delta) &\text{ for } \delta \leq 0
	\vspace{2mm}\\
	\dfrac{3}{2} + min\left\{\dfrac{3}{4r_L^*}, \dfrac{9}{6 r_L^* + 2}\right\} \delta &\text{ for } \delta \geq 0.
	\end{cases}
	\end{align*}
	
	Now consider $r_L^* = 1$.
	
	As $\delta = r_L - r_L^* = r_L -1$ and $r_L \in[0,1]$, $\delta$ must be less or equal to $0$.
	Furthermore, $\size(I_S)$ must be less or equal to $|I_L|/6$, as otherwise $r_L^*$ would be less than $1$.
	
	\begin{align*}
	PH_3 &=|B_{XL}| + |B_L| + |B_M| + |B_S|
	\end{align*}
	
	As above, each large item $a \in I_{XL}$ is packed into a separate bin and medium items are packed by twos.
	\begin{align*}
	|B_{XL}|&=|I_{XL}|\\
	|B_M|   &=\left\lceil \frac{|I_M|}{2} \right\rceil \leq \frac{|I_M|+1}{2}
	\end{align*}
	
	As pointed out above, $|B_L| \leq |I_L| + max\{0,6\; \delta \; \size(I_S) + 3\}$.
	As $\delta \leq 0$, $|B_L| \leq |I_L| + 3$.
	
	The argumentat used above to get a bound on $|B_S|$ applies here as well.
	\begin{align*}
	|B_S| & \leq \frac{3}{2} (1-r_L) \size(I_S) + 1\\
	& = \frac{3}{2} (-\delta) \size(I_S) + 1\\
	& \leq \frac{3}{2} (-\delta) \frac{|I_L|}{6} + 1\\
	& = \frac{|I_L|}{4}(-\delta) + 1\\
	& \leq \frac{OPT}{4}(-\delta) + 1
	\end{align*}
	
	In summary, this yields an upper bound on PH3.
	\begin{align*}
	PH_3 &\leq 
	\underbrace{\vphantom{\frac{1}{2}} |I_{XL}|}_{=|B_{XL}|}
	+ \underbrace{\frac{|I_M|+1}{2}}_{\geq |B_M|}
	+ \underbrace{\vphantom{\frac{1}{2}} |I_L| + 3}_{\geq |B_L|}
	+ \underbrace{\frac{OPT}{4}(-\delta) + 1}_{\geq |B_S|}\\
	&= |I_{XL}| + |I_L| + \frac{|I_M|}{2} + \frac{OPT}{4}(-\delta) + \frac{9}{2}
	\end{align*}
	
	For $|I_L| \geq |I_M|$ an estimate against $OPT$ can be formulated as follows using Lemma~\ref{l:lbound_OPT}.
	\begin{align*}
	PH_3 &\leq \underbrace{\vphantom{\frac{1}{1}} |I_{XL}| + |I_L|}_{\leq OPT} 
	+ \underbrace{\frac{|I_L|}{2}}_{\leq OPT/2} + \frac{OPT}{4}(-\delta) + \frac{9}{2}\\
	&\leq \left(\frac{3}{2} + \frac{1}{4} (-\delta) \right) OPT + \frac{9}{2}
	\end{align*}
	
	For $|I_L| < |I_M|$, $PH_3$ can be estimated against $OPT$ using another bound from Lemma~\ref{l:lbound_OPT}.
	
	\begin{align*}
	PH_3 &\leq |I_{XL}| + \frac{4 |I_L| + 2 |I_M|}{4} + \frac{OPT}{4}(-\delta) + \frac{9}{2}\\
	&< \underbrace{|I_{XL}| + \frac{3 |I_L| + 3 |I_M|}{4}}_{\leq 3/2 \; OPT} + \frac{OPT}{4}(-\delta) + \frac{9}{2}\\
	&\leq \left(\frac{3}{2} + \frac{1}{4}(-\delta)\right) OPT +  \frac{9}{2}
	\end{align*}

	This bound leads straight to the competitive ratio.
	\begin{align*}
	\Rightarrow \lim_{OPT \to \infty} \frac{PH_3}{OPT} 
	&\leq \lim_{OPT \to \infty} \left(\frac{3}{2} + \frac{1}{4} (-\delta) \right) \frac{OPT}{OPT} + \frac{9}{2 \; OPT}\\
	&= \frac{3}{2} + \frac{1}{4} (-\delta)\\
	&= \dfrac{3}{2} + min\left\{\dfrac{1}{4r_L^*}, \dfrac{3}{6 r_L^* + 2}\right\} (-\delta) , \;\delta \leq 0\qedhere
	\end{align*}
\end{proof}

\subsection{Tightness}
\begin{theorem}\label{t:tightness}
	
	For any $r_L, r_L^* \in [0,1]$, the asymptotic competitive ratio given in \autoref{t:cr} is tight.
\end{theorem}
{	\newcommand{\rep}{\mathop{\times}}
	\newcommand{\half}{\frac{1}{2} + \frac{\varepsilon}{2}}
	\newcommand{\thirdE}{\frac{1}{3} + \frac{\varepsilon}{2}}
	\newcommand{\sixth}{\frac{1}{6} - \varepsilon}
	\newcommand{\third}{\frac{1}{3} - 2 \varepsilon}
	\newcommand{\threeE}{3 \varepsilon}
	\newcommand{\twelveE}{12 \varepsilon}
	\newcommand{\halfT}{1/2 + \varepsilon/2}
	\newcommand{\thirdET}{1/3 + \varepsilon/2}
	\newcommand{\sixthT}{1/6 - \varepsilon}
	\newcommand{\thirdT}{1/3 - 2 \varepsilon}

	\begin{proof}
		Let $\langle a_1, a_2, \dots a_k \rangle \rep n$ with $n \in \mathbb{N}$ denote $n$ repetitions of the sequence $\langle a_1, a_2, \dots a_k \rangle$.
		
		Let $N \in \mathbb{N}$ and $\varepsilon = 1 / (12 N + 2)$.
		
		Let $I$ be a sequence consisting of the concatenated subsequences $I_S$, $I_M$ and $I_L$. 
		Let $I_S$ be a sequence consisting of the two interleaved subsequences $I_{SL}$ and $I_{LL}$.
		
		Let
		\begin{align*}
		I_L &= \left(\half\right) \rep n_L \text{ with } n_L = \lceil 4 r_L^* N \rceil \\
		I_M &= \left(\thirdE\right) \rep n_M \text{ with } n_M = 
		\begin{cases}
		0 &\text{ for } r_L^* \leq 1/3 \\ 
		\lfloor (6 r_L^* - 2) N \rfloor &\text{ for } r_L^* \geq 1/3
		\end{cases} \\
		I_{SS} &= \left(\third, \sixth, \sixth, \twelveE\right) \rep n_{SS} \text{ with } n_{SS} = \lceil n_{SS}' \rceil = \lceil (1 - r_L) N \rceil \\
		I_{SL} &= \left(\sixth, \threeE\right) \rep n_{SL} \text{ with } n_{SL} = \lceil n_{SL}' \rceil = \lceil 4 r_L N \rceil
		\end{align*}
		
		Let $I_S$ be interleaved in such a way that whenever PH3 would pack the next item in an L-bin, the next item in $I_S$ is the next item in $I_{SL}$, otherwise it is the next item in $I_{SS}$.
		If the corresponding subsequence contains no more items, the next item is taken from the other one.
		
		Note that the size of small items packed into each S-bin is exactly four times the size of small items packed into each L-bin.
		As $n_{SL} \geq n_{SL}'$ and $n_{SS} \geq n_{SS}'$ and
		
		\begin{align*}
		\frac{\size(I_{SL})}{\size(I_S)} \approx
		& \frac{n_{SL}'}{n_{SL}' + 4 n_{SS}'} = r_L.
		\end{align*}
		PH3 will pack at least $n_{SS}$ S-bins and $n_{SL}$ L-bins.
		The last S-bin packed by PH3 may contain items from $I_{SS}$ and the last L-bin might contain items from $I_{SL}$, depending on the differences $n_{SS} - n'_{SS}$ and $n_{SL} - n'_{SL}$.
		For the sake of simplicity, the up to two possible additional bins packed with small items are ignored, as they do not have impact on the competitive ratio due to the infinite length of the input sequences considered.
		
		\begin{figure}[h]
			\begin{tikzpicture}[scale=3]

\draw [<->] (-.15,0) -- (-.15,1) node [left,pos=.5] {1};

\draw [dashed](0,2/3) -- (0,1) -- (.5,1) -- (.5,2/3);
\draw (0,0) rectangle +(.5,2/3);
\foreach \i in {2/6,3/6} \draw (0,\i) -- +(.5,0);
\draw (.5,1/6) node [right] {$1/3 - 2 \varepsilon$};
\draw (.5,5/12) node [right] {$1/6 - \varepsilon$};
\draw (.5,7/12) node [right] {$1/6 - \varepsilon$};
\draw [<-] (.52,2/3+.01) -- +(.2,.1) node [right] {$12 \varepsilon$};
\draw (.25,0) node [below=.3cm] {S-bin};

\draw [dashed](1.5,2/3) -- (1.5,1) -- (2,1) -- (2,2/3);
\draw (1.5,0) rectangle +(.5,2/3);
\foreach \i in {2/6,4/6} \draw (1.5,\i) -- +(.5,0);
\foreach \i in {1/2,1/6}\draw (2,\i) node [right]{$1/3 + \varepsilon/2$};
\draw (1.75,0) node [below=.3cm] {M-bin};

\begin{scope}[shift={(-3.75,-1.5)}]
\draw [<->] (2.85,0) -- (2.85,1/3) node [pos=.5,left] {$1/3$};
\draw [<->] (2.85,1) -- (2.85,1/3) node [pos=.5,left] {$2/3$};
\draw [dashed](3,5/6) -- (3,1) -- (3.5,1) -- (3.5,5/6);
\draw [dashed](3,1/6) -- (3,1/3);
\draw [dashed](3.5,1/6) -- (3.5,1/3);
\draw (3,1/3) rectangle (3.5,5/6);
\draw (3.5,1/3+1/2/2) node [right] {$1/2+\varepsilon/2$};
\draw (3,0) rectangle (3.5,1/6);
\draw (3.5,1/12) node [right] {$1/6-\varepsilon$};
\draw [<-] (3.52,1/6+.01) -- +(.2,.1) node [right] {$3 \varepsilon$};
\draw (3.25,0) node [below=.3cm] {L-bin (a)};
\end{scope}

\begin{scope}[shift={(-2.25,-1.5)}]
\draw [dashed](3,5/6) -- (3,1) -- (3.5,1) -- (3.5,5/6);
\draw [dashed](3,1/3) -- (3,0) -- (3.5,0) -- (3.5,1/3);
\draw (3,1/3) rectangle (3.5,5/6);
\draw (3.5,1/3+1/2/2) node [right] {$1/2+\varepsilon/2$};
\draw (3.25,0) node [below=.3cm] {L-bin (b)};
\end{scope}

\begin{scope}[shift={(-.75,-1.5)}]
\draw [dashed](3,1/6) -- (3,1) -- (3.5,1) -- (3.5,1/6);
\draw [dotted](3,1/3) -- (3.5,1/3);
\draw (3,0) rectangle (3.5,1/6);
\draw (3.5,1/12) node [right] {$1/6-\varepsilon$};
\draw [<-] (3.52,1/6+.01) -- +(.2,.1) node [right] {$3 \varepsilon$};
\draw (3.25,0) node [below=.3cm] {L-bin (c)};
\end{scope}
\end{tikzpicture}
			\caption{The main types of bins packed by PH3. Dependent on $r_L$ and $r_L^*$, type (a) L-bins and type (b) or type (c) L-bins are packed.}
			\label{f:tightness_ssc_bin_types}
		\end{figure}
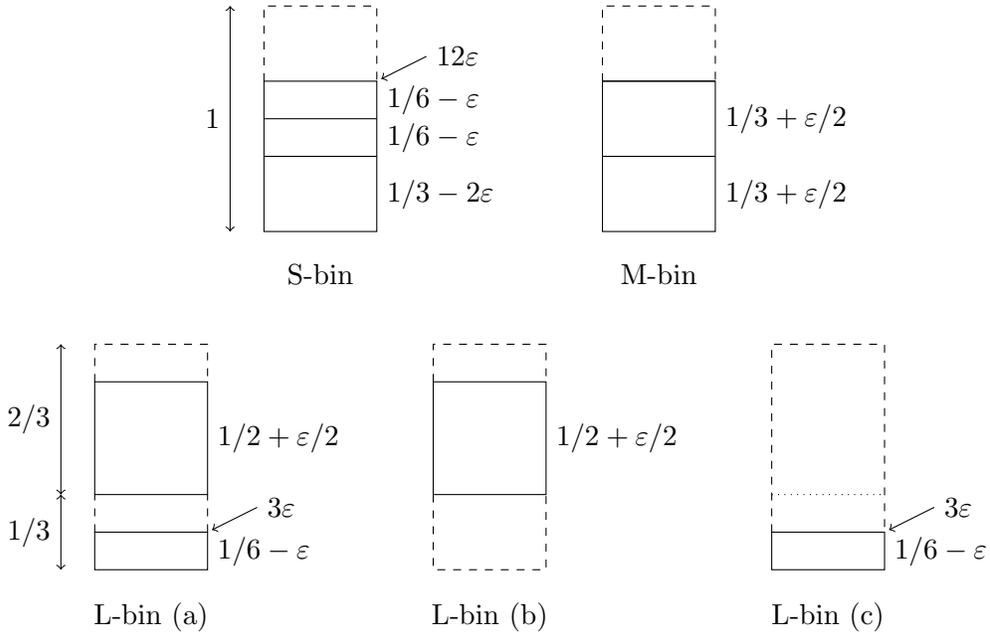
		
		When packing S-bins, PH3 packs the items $\thirdT$, $\sixthT$, $\sixthT$, $\twelveE$.
		The next item with size $\thirdT$ does no longer fit and a new bin is opened.
		As the subsequenc packed into each S-bin is repeated $n_{SS}$ times, $|B_S| = n_{SS}$.
		
		As items of size $\thirdET$ are packed by twos in M-bins, $|B_M| = \lceil n_M / 2 \rceil$.
		
		When packing the 1/3-sub-bins of the L-bins, PH3 packs one item of size $\sixthT$, then one item of size $\threeE$.
		The next item has size $\sixthT$.
		As $\sixthT + \threeE + \sixthT > 1/3$, a new bin will be opened.
		Accordingly, each L-bin $b \in B_L$ that is packed with small items will contain exactly the items $\sixthT$ and $\threeE$.
		As these two items occur $n_{SL}$ times each in $I_{SL}$, the number of these bins is $n_{SL}$.
		
		The items of size $\halfT$ are large and thus packed into the $2/3$-sub-bins of L-bins individually.
		As there are $n_L$ such items, the number of L-bins that contain a large item is $n_L$ as well.
		
		The total number of L-bins is the maximum of the number of L-bins that contain small items and L-bins that contain large items: $|B_L| = \max\{n_{SL}, n_L\}$.
		
		As the input sequence contains no extra large items, $|B_{XL}| = 0$.
		
		The packing of the bins is illustrated in \autoref{f:tightness_ssc_bin_types}.
		Overall, PH3 needs the following number of bins.
		\begin{align*}
		PH_3 &= |B_{XL}| + |B_L| + |B_M| + |B_L|\\
		&= 0 + \max\{n_{SL}, n_L\} + \left\lceil \frac{n_M}{2} \right\rceil + n_{SS}\\
		&\geq \max\{4 r_L N, 4 r_L^* N\} + \frac{n_M}{2} + n_{SS} \\
		&= \max\{r_L, r_L^*\} 4 N + \frac{n_M}{2} + (1 - r_L) N \\
		&= 3 r_L^* N + \max\{\delta, 0\} 4 N + \frac{n_M}{2} + N - \delta N 
		\end{align*}
		
		To get an upper bound on OPT, consider the algorithm FFD.
		FFD is an offline bin packing algorithm that sorts the items in decreasing order and puts each item in the first bin that has enough space left to fit the item \cite{doi:10.1137/0203025}.
		Note that although the order of the items is fixed, FFD can simulate an ordered list by assigning a bin to each item beforehand.
		
		\begin{figure}[h]
			\begin{tikzpicture}[scale=3]

\draw [<->] (-.15,0) -- (-.15,1) node [left,pos=.5] {1};

\draw (0,0) rectangle (.5,1);
\foreach \i in {1/2,5/6} \draw (0,\i) -- (.5,\i);
\draw (.5,1/4) node [right] {$1/2+\varepsilon/2$};
\draw (.5,1/2+1/6) node [right] {$1/3+\varepsilon/2$};
\draw (.5,1/2+1/3+1/12) node [right] {$1/6-\varepsilon$};
\draw (.25,0)  node [below] {$\underbrace{\hspace{1.5cm}}_{\in B_M}$};

\begin{scope}[shift={(1.5,0)}]
\draw (0,0) rectangle (.5,1);
\foreach \i in {1/2,5/6} \draw (0,\i) -- (.5,\i);
\draw (.5,1/4) node [right] {$1/2+\varepsilon/2$};
\draw (.5,1/2+1/6) node [right] {$1/3 - 2 \varepsilon$};
\draw (.5,1/2+1/3+1/12) node [right] {$1/6-\varepsilon$};
\end{scope}
\begin{scope}[shift={(3,0)}]
\draw (0,0) rectangle (.5,1);
\foreach \i in {1/2,4/6,5/6} \draw (0,\i) -- (.5,\i);
\draw (.5,1/4) node [right] {$1/2+\varepsilon/2$};
\draw (.5,1/2+1/12) node [right] {$1/6-\varepsilon$};
\draw (.5,1/2+1/6+1/12) node [right] {$1/6-\varepsilon$};
\draw (.5,1/2+1/3+1/12) node [right] {$1/6-\varepsilon$};
\end{scope}
\draw (2.5,0)  node [below] {$\underbrace{\hspace{6cm}}_{\in B_L}$};

\begin{scope}[shift={(.75,-1.5)}]
\draw [<->] (-.15,0) -- (-.15,1) node [left,pos=.5] {1};
\draw (0,0) rectangle (.5,1);
\foreach \i in {1/3,2/3} \draw (0,\i) -- (.5,\i);
\draw (.5,1/6) node [right] {$1/3-2\varepsilon$};
\draw (.5,3/6) node [right] {$1/3-2\varepsilon$};
\draw (.5,5/6) node [right] {$1/3-2\varepsilon$};
\end{scope}

\begin{scope}[shift={(2.25,-1.5)}]
\draw (0,0) rectangle (.5,1);
\foreach \i in {1/6,2/6,3/6,4/6,5/6} \draw (0,\i) -- (.5,\i);
\foreach \i in {1/12,3/12,5/12,7/12,9/12,11/12} \draw (.5,\i) node [right] {$1/6-\varepsilon$};
\end{scope}

\draw (1.75,-1.5) node [below] {$\underbrace{\hspace{6cm}}_{\in B_S}$};

\end{tikzpicture}
			\caption{The main types of bins packed by FFD.}
			\label{f:tightness_ffd}
		\end{figure}
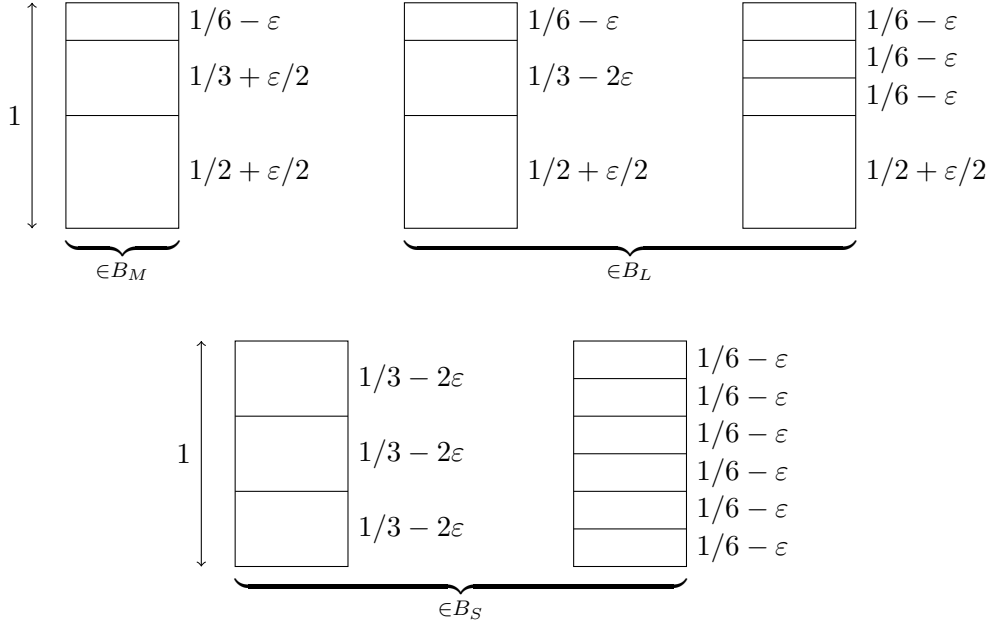
		
		FFD will first pack each large item in a separate bin, opening $n_L$ bins.
		
		The medium items are packed one each in the first $n_M$ bins alongside the large items.
		Note that there are $n_L = 4 r_L^* N$ large items and $n_M = \lfloor (6 r_L^* - 2) N \rfloor$ medium items.
		As $r_L^* \leq 1$, $n_L \geq n_M$ and thus each medium item can be packed with a large item.
		Let $B_L$ the set of bins containing only a large item and $B_M$ the set of bins containing a large and a medium item.
		Note that $|B_M| = n_M$ and $|B_L| = n_L - n_M$.
		
		Items of size $\thirdT$ are packed into the bins $B_L$ one each.
		If there are more such items than there are bins in $B_L$, the remaining items are packed by threes.
		
		Afterwards, the items of size $\sixthT$ are packed into the bins $B_M$ and $B_L$, items that do not fit are packed into additional bins.
		
		Each bin in $B_M$ then contains, besides to the large and the medium item, exactly one item of size $\sixthT$.
		
		Each bin $b \in B_L$ is packed either with the items $\{\halfT, \thirdT, \sixthT\}$ or the items $\{\halfT, \sixthT, \sixthT, \sixthT\}$.
		In either case, the size of the packed small items is $\size(b \cap S) = 1/2 - 3 \varepsilon$.
		
		Let $size_S$ denote the total size of all items of size $\thirdT$ and $\sixthT$.
		Let $size_S'$ denote the total size those items of size $\thirdT$ and $\sixthT$ that are packed into additional bins.
		\begin{align*}
		size_S' 
		&= size_S - |B_L| \left(\frac{1}{2} - 3 \varepsilon\right) - |B_M| \left( \sixth \right) \\
		&= n_{SS} \left(\third\right) + (n_{SL} + 2 n_{SS} - n_M) \left(\sixth\right) - (n_L - n_M) \left(\frac{1}{2} - 3 \varepsilon\right) \\
		&= (n_{SL} + 4 n_{SS} - 3 n_L + 2 n_M) \left(\sixth\right) 
		\end{align*}
		
		When packing the items of size $\thirdT$ and $\sixthT$ in additional bins, FFD packs three items of size $\thirdT$ or six items of size $\sixthT$ in each bin except one bin, which may contain items of both sizes, and the last bin, which may contain fewer items.
		Each of these bins, except for the last hast size $1 - 6 \varepsilon$.
		Let $B_S$ denote the set of these bins.
		\begin{align*}
		B_S &= \left\lceil \frac{size_S'}{1 - 6 \varepsilon} \right\rceil\\
		&= \left\lceil \frac{1}{1 - 6 \varepsilon} \left(\sixth\right) (n_{SL} + 4 n_{SS} - 3 n_L + 2 n_M) \right\rceil\\
		&\leq \frac{n_{SL} + 4 n_{SS} - 3 n_L + 2 n_M + 6}{6}
		\end{align*}
		
		For packing the items of size $\twelveE$ and the items of size $\threeE$, at most one additional bin is used, as
		\begin{align*}
		n_{SS} (12 \varepsilon) + n_{SL} (3 \varepsilon)
		&\leq (12 n_{SS}' + 1) \varepsilon + (3 n_{SL} + 1) \varepsilon\\
		&= (12 - 12 r_L) N \varepsilon + 12 r_L N \varepsilon + 2 \varepsilon\\
		&= (12 N + 2) \varepsilon\\
		&= 1
		\end{align*}
		
		\autoref{f:tightness_ffd} shows a visualization of the main types of bins packed by FFD before items of size $\twelveE$ or smaller are packed.
		As FFD provides a feasible solution to the offline bin packing problem, this also yields a lower bound on the optimal solution OPT.
		
		\begin{align*}
		OPT \leq FFD
		&\leq |B_L| + |B_M| + |B_S| + 1\\
		&\leq n_L + \frac{n_{SL} + 4 n_{SS} - 3 n_L + 2 n_M + 6}{6} + 1\\
		&= \frac{3 n_L + n_{SL} + 4 n_{SS} + 2 n_M + 12}{6}\\
		&= \frac{3 \lceil 4 r_L^* N \rceil + \lceil 4 r_L N \rceil + 4 \lceil (1 - r_L) N \rceil + 2 n_M + 12}{6} \\
		&\leq \frac{12 r_L^* N + 4 r_L N + 4 (1 - r_L) N + 2 n_M + 20}{6} \\
		&= \frac{(12 r_L^* + 4) N + 2 n_M + 20}{6}
		\end{align*}

		Consider $r_L^* \geq 1/3$.
		
		\begin{align*}
		n_M &= \lfloor (6 r_L^* - 2) N \rfloor\\
		OPT &\leq  \frac{(12 r_L^* + 4) N + 2 \lfloor(6 r_L^* -2) N \rfloor + 20}{6}\\
		&\leq \frac{24 r_L^* N + 20}{6} < 4 r_L^* N + 4\\
		PH_3 &\geq 3 r_L^* N + \max\{\delta, 0\} 4 N + \frac{\lfloor(6 r_L^* - 2) N \rfloor}{2} + N - \delta N \\
		&\geq 3 r_L^* N + \max\{\delta, 0\} 4 N + (3 r_L^* - 1) N + N - \delta N - 1\\
		&=    6 r_L^* N + \max\{\delta, 0\} 4 N - \delta N - 1
		\end{align*}
		
		Note that for $r_L^* \geq 1/3$, the following two equations hold.
		
		\begin{align*}
		min\left\{\dfrac{1}{4r_L^*}, \dfrac{3}{6 r_L^* + 2}\right\} &=\dfrac{1}{4r_L^*}\\
		min\left\{\dfrac{3}{4r_L^*}, \dfrac{9}{6 r_L^* + 2}\right\} &=\dfrac{3}{4r_L^*}
		\end{align*}
		
		Note that for $OPT \to \infty$, $N \to \infty$ also holds.
		Thus, for $N \to \infty$, a lower bound on the competitive ratio for $r_L^* \geq 1/3$ can be given.
		
		\begin{align*}
		\lim_{OPT \to \infty} \frac{PH_3}{OPT} 
		&\geq \lim_{N \to \infty} \frac{6 r_L^* N + \max\{\delta, 0\} 4 N - \delta N - 1}{4 r_L^* N + 4}\\
		&= \frac{3}{2} + \frac{4 \max\{\delta, 0\} - \delta}{4 r_L^*}\\
		&=\begin{cases}
		\dfrac{3}{2} + \dfrac{3}{4 r_L^*} \delta     &\text{ for } \delta \geq 0 
		\vspace{2mm}\\
		\dfrac{3}{2} + \dfrac{1}{4 r_L^*} (- \delta) &\text{ for } \delta \leq 0 \\
		\end{cases}\\
		&=R_{PH3}^\infty \text{ for } r_L^* \geq 1/3
		\end{align*}
		
		Now consider $r_L^* \leq 1/3$.
		
		\begin{align*}
		n_M &= 0\\
		OPT &\leq  \frac{(12 r_L^* + 4) N + 20}{6}\\
		PH_3 &\geq 3 r_L^* N + \max\{\delta, 0\} 4 N + N - \delta N 
		\end{align*}
		
		Note that for $r_L^* \leq 1/3$, the following two equations hold.
		
		\begin{align*}
		min\left\{\dfrac{1}{4r_L^*}, \dfrac{3}{6 r_L^* + 2}\right\} &=\dfrac{3}{6r_L^* + 2}\\
		min\left\{\dfrac{3}{4r_L^*}, \dfrac{9}{6 r_L^* + 2}\right\} &=\dfrac{9}{6r_L^* + 2}
		\end{align*}
		
		As above, for $N \to \infty$, a lower bound on the competitive ratio for $r_L^* \leq 1/3$ can be given .
		
		\begin{align*}
		\lim_{OPT \to \infty} \frac{PH_3}{OPT} 
		&\geq \lim_{N \to \infty} 6\frac{3 r_L^* N + \max\{\delta, 0\} 4 N + N - \delta N}{(12 r_L^* + 4) N + 20}\\
		&= \frac{9 r_L^* + 12 \max\{\delta, 0\} + 3 - 3 \delta}{6 r_L^* + 2}\\
		&= \frac{3}{2} + \frac{12 \max\{\delta, 0\}- 3 \delta}{6 r_L^* + 2}\\
		&=\begin{cases}
		\dfrac{3}{2} + \dfrac{9}{6 r_L^* + 2} \delta     &\text{ for } \delta \geq 0 
		\vspace{2mm}\\
		\dfrac{3}{2} + \dfrac{3}{6 r_L^* + 2} (- \delta) &\text{ for } \delta \leq 0 \\
		\end{cases}\\
		&=R_{PH3}^\infty \text{ for } r_L^* \leq 1/3 \qedhere
		\end{align*}
	\end{proof}
}
\section{Parallel \textsc{Predictive Harmonic$_3$}}
\subsection{Competitive Ratio for PH3 as 1-Copy Online Algorithm}
\newcommand{\rph}{\mathop{r_{PH3}^\infty}}

To optimize the performance for PH3 as a 1-copy algorithm, we determine the optimal value for $r_L$ with respect to 
minimizing the asymptotic competitive ratio over all $r_L^* \in [0,1]$.

\begin{lemma}[Monotonicity of competitive ratio of PH3]
\label{l:monotony}
For any fixed $r_L\in[0,1]$, the competitive factor is monotonically decreasing for $r_L^*\in [0,r_L]$ and monotonically increasing for $r_L^*\in [r_L,1]$.
\ignore{
	For any fixed $r_L \in [0,1]$ and $\delta \leq 0$, the competitive ratio of PH3 increases with $r_L^*$ increasing and $\delta$ decreasing.
	For any fixed $r_L \in [0,1]$ and $\delta \geq 0$, the competitive ratio of PH3 increases with $r_L^*$ decreasing and $\delta$ increasing.
}
\end{lemma}

\begin{proof}

Assume $r_L$ to be fixed.
Let $r_{+,<}, r_{-,<}: [0,1/3] \to \mathbb{R}$ and $r_{+,>}, r_{-,>}: [1/3,1] \to \mathbb{R}$ with

\begin{align*}
r_{-,<}(r_L^*) &= \frac{3}{2} + \frac{3}{6 r_L^* + 2} (-\delta) &= \Rph \text{ for } \delta \leq 0, r_L^* \leq \frac{1}{3}\\
r_{-,>}(r_L^*) &= \frac{3}{2} + \frac{1}{4 r_L^*} (-\delta)     &= \Rph \text{ for } \delta \leq 0, r_L^* \geq \frac{1}{3}\\
r_{+,<}(r_L^*) &= \frac{3}{2} + \frac{9}{6 r_L^* + 2} \delta    &= \Rph \text{ for } \delta \geq 0, r_L^* \leq \frac{1}{3}\\
r_{+,>}(r_L^*) &= \frac{3}{2} + \frac{3}{4 r_L^*} \delta        &= \Rph \text{ for } \delta \geq 0, r_L^* \geq \frac{1}{3}\\
\end{align*}

Consider the derivative of $r_{-,<}$ and $r_{-,>}$.

\begin{align*}
\frac{\partial}{\partial r_L^*} r_{-,<}(r_L^*) 
&= \frac{\partial}{\partial r_L^*} \left(\frac{3}{2} + \frac{3}{6 r_L^* + 2} (-\delta)\right)\\
&= \frac{\partial}{\partial r_L^*} \left( \frac{3 (r_L^* - r_L)}{6 r_L^* + 2} \right)\\
&= \frac{18 r_L + 6}{(6 r_L^* + 2)^2} \geq 0 \text{ for } 0 \leq r_L \leq r_L^* \leq \frac{1}{3}\\
\frac{\partial}{\partial r_L^*} r_{-,>}(r_L^*) 
&= \frac{\partial}{\partial r_L^*} \left(\frac{3}{2} + \frac{1}{4r_L^*} (-\delta)\right)\\
&= \frac{\partial}{\partial r_L^*} \left(\frac{r_L^* - r_L}{4 r_L^*} \right)\\
&= \frac{r_L}{4 (r_L^*)^2} \geq 0 \text{ for } 0 \leq r_L \leq r_L^* \text{ and } \frac{1}{3} \leq r_L^* \leq 1
\end{align*}

As the derivatives of $r_{-,<}$ and $r_{-,>}$ are both non-negative in their respective domains, they are both monotonically increasing.
Because $r_{-,<}(\frac 1 3) = r_{-,>}(\frac 1 3)$, we conclude that the competitive ratio is monotonically increasing for $r_L^*\in [r_L,1]$. 

Now consider the derivative of $r_{+,<}$ and $r_{+,>}$.

\begin{align*}
\frac{\partial}{\partial r_L^*} r_{+,<}(r_L^*) 
&= \frac{\partial}{\partial r_L^*} \left(\frac{3}{2} + \frac{9}{6 r_L^* + 2} \delta\right)\\
&= \frac{\partial}{\partial r_L^*} \left( \frac{9 (r_L - r_L^*)}{6 r_L^* + 2} \right)\\
&= \frac{-54 r_L - 18}{(6 r_L^* + 2)^2} \leq 0 \text{ for } r_L^* \leq r_L \leq 1 \text{ and } 0 \leq r_L^* \leq \frac{1}{3}\\
\frac{\partial}{\partial r_L^*} r_{+,>}(r_L^*) 
&= \frac{\partial}{\partial r_L^*} \left(\frac{3}{2} + \frac{3}{4r_L^*} \delta\right)\\
&= \frac{\partial}{\partial r_L^*} \left(\frac{3 (r_L - r_L^*)}{4 r_L^*} \right)\\
&= \frac{-3 r_L}{4 (r_L^*)^2} \leq 0 \text{ for } \frac{1}{3} \leq r_L^* \leq r_L \leq 1
\end{align*}

As the derivatives of $r_{+,<}$ and $r_{+,>}$ are both non-positive in their respective domains, they are both monotonically decreasing.
Because $r_{+,<}(\frac 1 3)=r_{+,>}(\frac 1 3)$, we conclude that the competitive ratio is monotonically decreasing for $r_L^*\in[0,r_L]$.
\end{proof}

Because of Lemma~\ref{l:monotony}, the competitive ratio does not decrease with $r_L^*$ increasing for $\delta \leq 0$.
Thus, as an upper bound on the competitive ratio for $\delta \leq 0$, only the competitive ratio for $r_L^* = 1$ has to be considered.
\begin{align*}
\Rph &\leq \frac{3}{2} + \frac{1}{4} (-\delta) \text{ for } \delta \leq 0\\
&= \frac{3}{2} + \frac{1}{4} (1 - r_L)\\
&= \frac{7}{4} - \frac{r_L}{4}
\end{align*}

For $\delta \geq 0$, the competitive ratio does not decrease with $r_L^*$ decreasing.
In this case, the competitive ratio for $r_L^* = 0$ is an upper bound on the competitive ratio.
\begin{align*}
\Rph &\leq \frac{3}{2} + \frac{9}{2} \delta \text{ for } \delta \geq 0\\
&= \frac{3}{2} + \frac{9}{2} (r_L - 0)\\
&= \frac{3}{2} + \frac{9}{2} r_L
\end{align*}

At the same time, these values are lower bounds on the overall competitive ratio.
Given these bounds, this linear program can be formulated to minimize the competitive ratio:
\begin{align*}
\mathop{\text{Minimize }} \Rph &\\
\mathop{\text{Subject to }}
  \Rph &\geq \frac{7}{4} - \frac{r_L}{4}\\
  \Rph &\geq \dfrac{3}{2} + \dfrac{9}{2} r_L\\
r_L & \geq 0\\
r_L & \leq 1\\
\end{align*}

The optimal solution for this linear program is $r_L=1/19$ and $\Rph = 33/19 < 1.7369$.
\autoref{f:ph3_universe} shows the asymptotic competitive ratio of PH3 over $r_L^*$ for $r_L = 1/19$.

\begin{figure}[h]
	\centering
\input{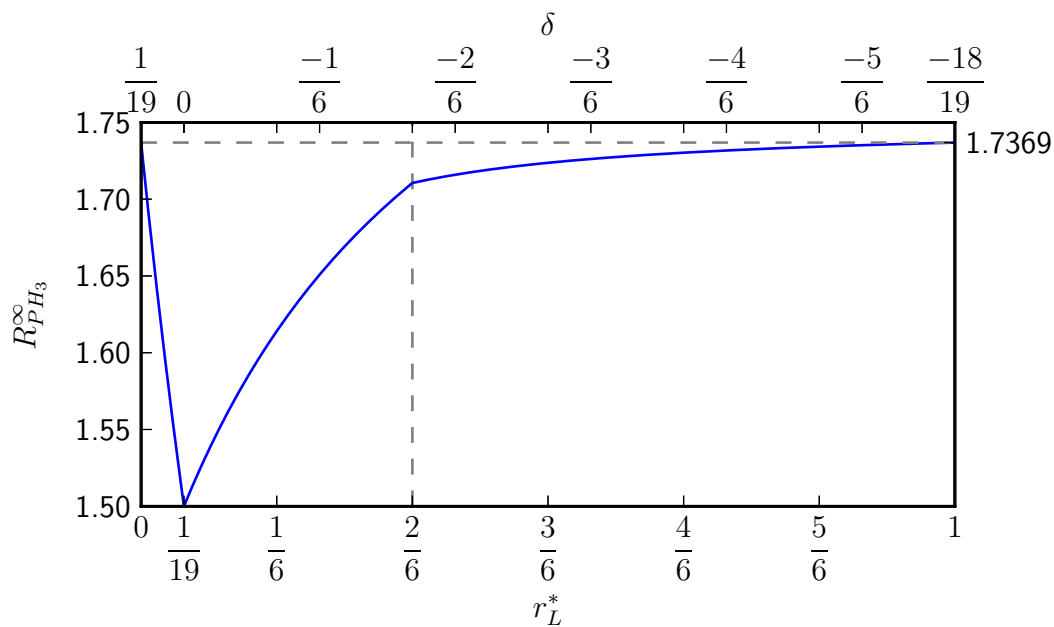}
\caption{Competitive ratio of the optimal 1-copy PH3 algorithm dependent on $r_L^*$ for a fixed $r_L$.}
\label{f:ph3_universe}
\end{figure}

Compared to other known algorithms for online bin packing, PH3 is not a good choice for worst-case behavior.
Among the classical algorithms, only NF and WF, both of which are 2-competitive, are worse than PH3.
Any AAF algorithm achieves an asymptotic competitive ratio $R_{AAF}^\infty = 1.7$ \cite{Csirik1998} and thus performs slightly better than PH3.
The best-performing online algorithm for bin packing currently known, \textsc{Son Of Harmonic}, is 1.5816-competitive and thus clearly superior to PH3 \cite{DBLP:journals/corr/HeydrichS15}.

However, if we know in advance that $r_L^*$ is restricted to some interval $I_r = [a,b] \subset [0,1]$, the above argument can be used to prove a better competitive ratio.

\subsection{Competitive Ratio for PH3 as \boldmath{$k$}-Copy Online Algorithm}

PH3's property of achieving a better competitive ratio for $r_L^*$ being further restricted can be used to create a set of $k \in \mathbb{N}$ algorithms achieving a better competitive ratio.
For this purpose, the interval $[0,1]$ is split into $k$ sub-intervals $I_1, \dots, I_k \subset [0,1]$ with $\cup_{i \in \{1, \dots, k\}} I_i = [0,1]$.
Each interval $I_i$ is covered by one instance of the algorithm PH3 $A_i$, such that $A_i$ achieves a targeted competitive ratio $R \in (3/2, 33/19)$ for $r_L^* \in I_i$.

$R$ is restricted to $(3/2, 33/19)$, because any competitive ratio above or equal to $33/19$ can be achieved with the instance of PH3 shown above, and k-copy PH3 cannot achieve a competitive ratio of $3/2$ or less with finitely many algorithms.

To calculate the number $k$ of algorithms needed to achieve a given competitive ratio $R$, the following iterative approach can be used.

Let $\mathcal{A}$ be a set of algorithms. Initially, $\mathcal A := \emptyset$.
We initialize our iterative approach with $i=0$ and set $r_{max}^0=0$.
Then, while $r_{max}^i < 1$, we increase $i$ by one and we compute three values $r_{min}^{i}$, $r_L^{i}$ and $r_{max}^i$.
With these three values we can define algorithm $A_i$ for which $r_L^i$ denotes the value of $r_L$, $r^i_{min}$ denotes the minimal and $r^i_{max}$ denotes the maximal value for $r_L^*$ for which $A_i$ is still $R$-competitive.
By Lemma~\ref{l:monotony}, $A_i$ will be $R$-competitive for the interval $[r^i_{min}, r^i_{max}]$.
All three values are computed as follows.
We set $r^i_{min}=r^{i-1}_{max}$.
Given $r^i_{min}$, $r_L^i$ can be computed:

If $r^i_{min} \leq 1/3$, 
we have $R=\frac{3}{2} + \frac{9}{2+6 r^i_{min}} (r_L^i - r^i_{min})$.
Solving this equation for $r_L^i$ we get $r_L^i = r^i_{min} + \left(R - \frac{3}{2}\right) \left(\frac{2+6 r^i_{min}}{9}\right)$.
If $r^i_{min} \geq 1/3$, 
we have $R = \frac{3}{2} + \frac{3}{4 r^i_{min}} (r_L^i - r^i_{min})$. 
Solving this equation for $r_L^i$ yields $r_L^i = r^i_{min} + \left(R - \frac{3}{2}\right) \left(\frac{4 r^i_{min}}{3}\right)$.

Having $r_L^i$, we can compute $r_{max}^i$. 
Because the competitive ratio is the minimum of two values, 
we get two candidates $r^i_{max,1}$ and $r^i_{max,2}$ for $r_{max}^i$. 
We can take the maximum of those two candidates, 
i.e., $r_{max}^i = \max( r^i_{max,1}, r^i_{max,2})$, 
because it is sufficient to be $R$-competitive in one case.
In the first case ($\frac {3}{6r_L^*+2} < \frac 1 {4r_L^*}$) we obtain $r_{max, 1}^i = \frac{3 r_L^i - 3 + 2R}{12 -6 R}$ and in the second case we get $r_{max, 2}^i = \frac{r_L^i}{7 -4 R}$.

Now consider the case when $r_{max}^i \geq 1$.  
Because each algorithm $A_\ell$ with $1\leq \ell\leq i$ is $R$-competitive for the interval $[r^\ell_{min},r^\ell_{max}]=[r^{\ell-1}_{max},r^{\ell}_{max}]$ with $r^0_{min} = 0$, 
there is an algorithm $A_m$ for any $r_L^*\in [0,1]$ that is $R$-competitive.
Therefore, we have a $i$-copy online algorithm for bin packing achieving the competitive factor $R$.

\begin{figure}[t]
	\resizebox{\columnwidth}{!}{\input{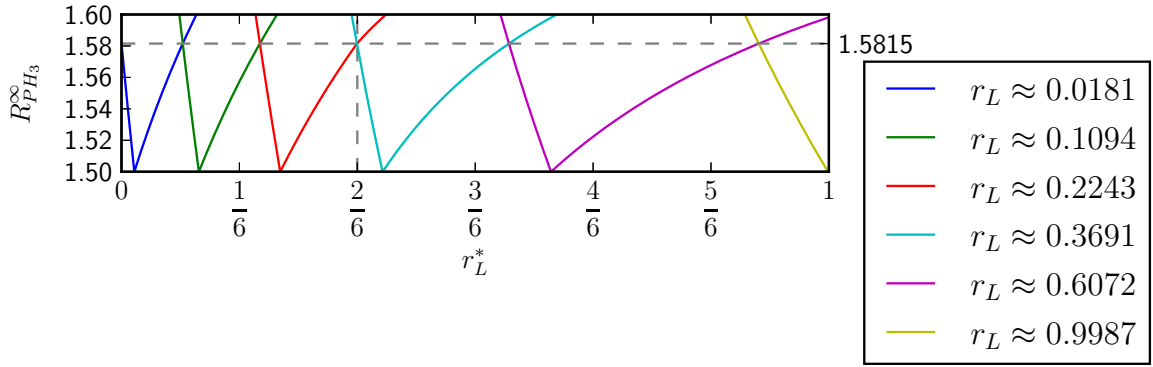}}
	\caption{6-copy PH3 beats the best 1-copy online algorithm known to date, achieving an asymptotic competitive ratio $\Rph < 1.5815$.}
	\label{f:ph3_1.5815}
\end{figure}

Following this method, we see that $k=6$ algorithms are sufficient to guarantee a competitive ratio $R=1.5815$.
This beats the currently best 1-copy online algorithm \textsc{Son Of Harmonic} with a competitive ratio of $1.5816$.
\autoref{f:ph3_1.5815} shows the competitive ratio achieved by the individual algorithms over $r_L^* \in [0,1]$ for $R=1.5815$.
Note that $1.5815$ is not the best competitive ratio achievable by 6-copy PH3, as shown below in Figure~\ref{f:ph3_cmp}.
Using $k=12$ algorithms, a competitive ratio $R=1.5402 < 1.5403$ can be achieved, beating the highest known lower bound for 1-copy online algorithms.

To compute the best competitive ratio achievable by $k \in \mathbb{N}$ algorithms, 
we use binary search on $R$ starting in the interval $[3/2, 33/19]$ 
and test in each iteration if we can guarantee $R$-competitiveness with at most $k$ algorithms.
\autoref{f:ph3_cmp} shows the best competitive ratios achievable by k-copy PH3.

\begin{figure}[h]
	\centering
	\input{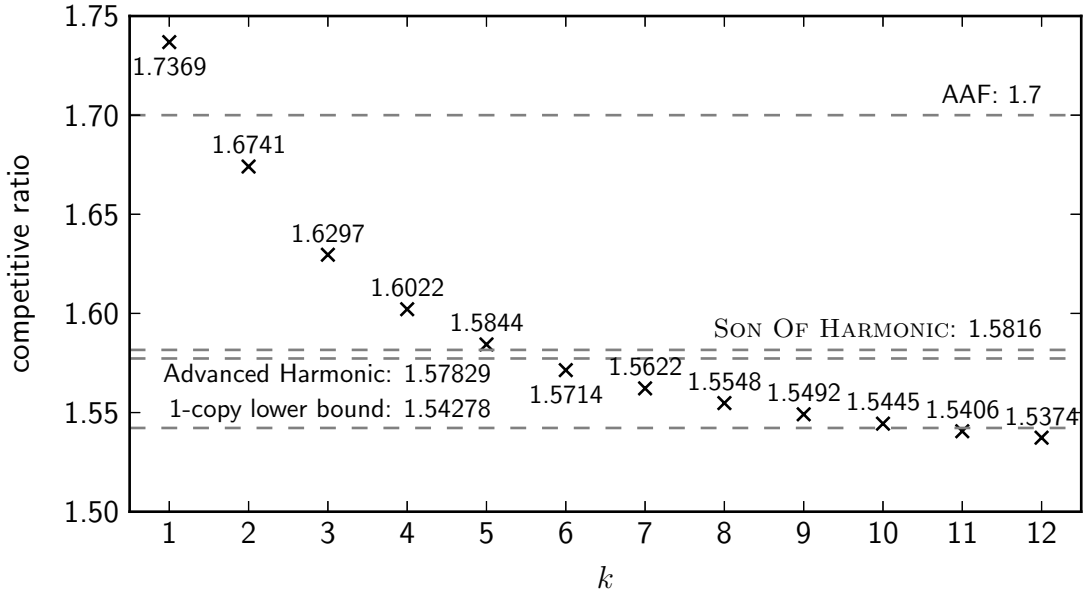}
	\caption{k-copy PH3 performance dependent on $k$.}
	\label{f:ph3_cmp}
\end{figure}

\subsection{Comparison to Related Algorithms}

Because $k$-copy online algorithms can be translated to an online algorithm with advice and vice versa (see Lemma~\ref{l:kcopy_advice}), it seems natural to compare these two variants, even though $k$-copy allows a more precise analysis on the competitive ratio.
In this subsection we compare our algorithm to the best known online algorithm with constant advice, namely \textsc{RedBlue} introduced by Angelopoulos et al. \cite{Angelopoulos2015}.
Their second algorithm is $1.47012$-competitive (and thus beats our algorithm), but the amount of advice needed by this algorithm is too large.
As the focus of k-copy algorithms is to provide good solutions for small k, it is reasonable to  only compare k-copy PH3 to \textsc{RedBlue}.

\begin{table}[t]
	\centering
	\begin{tabular}{|l|l|l|l|}
\hline
Advice in bits & $k$ & $R_{\text{\textsc{RedBlue}}}^\infty$ & $\Rph$ \\
\hline
4 & 16 & 3.3750 & 1.5305 \\
5 & 32 & 2.8258 & 1.5155 \\
6 & 64 & 2.4375 & 1.5078 \\
7 & 128 & 2.1629 & 1.5040 \\
8 & 256 & 1.9688 & 1.5020 \\
9 & 512 & 1.8315 & 1.5010 \\
10 & 1024 & 1.7344 & 1.5005 \\
11 & 2048 & 1.6657 & 1.5003 \\
12 & 4096 & 1.6172 & 1.5002 \\
13 & 8192 & 1.5829 & 1.5001 \\
14 & 16384 & 1.5586 & 1.5001 \\
15 & 32768 & 1.5414 & 1.5001 \\
16 & 65536 & 1.5293 & 1.5001 \\
\hline
\end{tabular}

	\caption{Comparison of the performance of k-copy PH3 and \textsc{RedBlue}.}
	\label{t:cmp_ph3_redblue}
\end{table}

\autoref{t:cmp_ph3_redblue} shows a comparison between \textsc{RedBlue} and $k$-copy PH3 for small amounts of advice.
The competitive ratios given are rounded up to the fourth decimal place.
The competitive ratios for \textsc{RedBlue} are computed using the upper bound on the competitive ratio $1.5 + 15/(2^{\ell/2+1})$.
The competitive ratios for $k$-copy PH3 are calculated using binary search as described above.

\autoref{t:cmp_ph3_redblue} clearly shows the advantage of $k$-copy PH3 over \textsc{RedBlue} for few bits of advice.
With as few as $5$ bits of advice, or $k = 32$, $k$-copy PH3 achieves a better competitive ratio than \textsc{RedBlue} with 16 bits of advice, which corresponds to $k=65536$ algorithms when used as $k$-copy algorithm.

Although \textsc{RedBlue} and $k$-copy PH3 work in a similar way, $k$-copy PH3 achieves a better competitive ratio due to the more precise analysis of the intervals for $r_L^*$, in which each algorithm achieves the competitive ratio.
By avoiding overlaps in these intervals, fewer algorithms are needed.

On the other hand, \textsc{RedBlue} 
simply splits an interval for its parameter $\beta$ evenly into $2^{\ell/2}$ intervals; translated into a $k$-copy setting, this leads
to overlaps in the intervals covered by each algorithm.

\section{Conclusion}
We studied the concept of parallel online algorithms for the Bin Packing Problem.
We developed a $k$-copy online algorithm named PH3 and showed that 
PH3 has an asymptotic competitive ratio of $1.5$ for large $k$; in particular, $k=11$ suffices
to break through the lower bound of a single online algorithm. We also considered
the relationship to online algorithms with advice and achieved a considerable improvement
compared to a previous algorithm.

There are various directions for future work. We saw that
	PH3 is $(1.5+\varepsilon)$-competitive if $\frac {|I_L|}{6\size(I_S)}\leq 1$, i.e., when there is a 
surplus of small items.
	If there are too few small items, PH3 is asymptotically $(1.5+\varepsilon)$-competitive.
	Can we make better use of the second case for an improvement? 
Can we guarantee an absolute competitive ratio of $1.5(+\varepsilon)$?
	
	How does the asymptotic competitive ratio of PH3 depend on $k$?
	It seems to be something like $\frac 3 2 + O\left(\frac {1}{k+\log_2(k+1)}\right)$.
	Translated to an online algorithm with $\ell$ bits of advice, this would yield an asymptotic competitive ratio of 	
	$\frac 3 2 + O\left(\frac 1 {2^\ell + \ell}\right)$.

We also believe that the concept of $k$-copy algorithms is useful
for a wide range of other problems.

\bibliography{bibliography}
\bibliographystyle{abbrv}

\end{document}